\documentclass[runningheads]{llncs}
\usepackage{hyperref}
\usepackage{amssymb}
\newfont{\fsc}{eusm10}      

\renewcommand{\phi}{\varphi}
\newcommand{\plat}[1]{\raisebox{0pt}[0pt][0pt]{#1}}     
\newcommand{\inp}{\mathop\in}                           

\newcommand{\setof}[2]{\{ \, #1 \, \mid \, #2 \, \}}
\def \BISI{\sim}
\def \WBISI{\approx}
\def \SI{\prec}
\def \WSI{\precsim}
\newcommand{\Act}{{\sf Act}}
\newcommand{\Var}{{\sf Var}}
\newcommand{\Env}{{\sf Env}}
\def \fv{{\it fv}}
\def \CR{{\cal R}}
\def \cL{{\cal L}}
\def \cF{{\cal F}}
\def \true{{\tt true}}
\def \false{{\tt false}}
\def \Op#1{[\![ #1 ]\!]}
\newcommand{\diam}[1]{\langle#1\rangle}
\newcommand{\boxm}[1]{[#1]}
\newcommand{\refuse}[1]{\mathop{\mathbf{ref}(#1)}}
\newcommand{\pow}[1]{\mathop{\mbox{\fsc P}} ({#1})   }
\newcommand{\dist}[1]{\mathop{\mbox{$\cal D$}} ({#1})   } 
\newcommand{\pdist}[1]{\overline{#1}  } 
\newcommand{\support}[1]{\lceil{#1}\rceil}
\newcommand{\lift}[1]{\mathrel{#1}^\dag}
\newcommand{\aRel}{\ensuremath{\mathrel{{\cal R }}}}
\newcommand{\forsim}{\mathrel{\lhd_{\raisebox{-.1em}{\tiny\it S}}}}
\newcommand{\failsim}{\mathrel{\lhd_{\raisebox{-.1em}{\tiny\it FS}}}}
\newcommand{\sset}[1]{\{ {#1}  \}  } 
\newcommand{\ar}[1]{\mathrel{\plat{$\buildrel #1 \over \longrightarrow$}}}
\newcommand{\dar}[1]{\mathrel{\plat{$\buildrel #1 \over \Longrightarrow$}}}
\newcommand{\darhat}[1]{\mathrel{\plat{$\buildrel \hat{#1} \over \Longrightarrow$}}}
\newcommand{\arh}[1]{\mathrel{\buildrel #1 \over \longrightarrow}}
\newcommand{\darhhat}[1]{\mathrel{\buildrel \hat{#1} \over \Longrightarrow}}
\newcommand{\nar}[1]{\hspace{6pt}\not\hspace{-6pt}\ar{#1\;}}

\begin{document}

\title{Characterising Probabilistic Processes Logically}
\author{Yuxin Deng\inst{1}\fnmsep\inst{2}%
\thanks{\hspace{-2pt}Deng was supported by the National Natural
Science Foundation of China
\mbox{(\hspace{-.5pt}60703033\hspace{-.5pt})\hspace{-.5pt}}.}
 \and Rob van Glabbeek\inst{3}\fnmsep\inst{4}}
\institute{Dept. Comp. Sci. \& Eng. and MOE-Microsoft Key Lab for Intell. Comp. \& Syst., Shanghai Jiao Tong University, China
\and State Key Lab of Comp. Sci., Inst. of Software, Chinese Academy of Sciences
 \and NICTA, Australia \and University of New South Wales, Australia} \maketitle

\begin{abstract}
In this paper we work on (bi)simulation semantics of processes that
exhibit both nondeterministic and probabilistic behaviour.
We propose a probabilistic extension of the modal mu-calculus and show
how to derive characteristic formulae for various simulation-like
preorders over finite-state processes without divergence.
In addition, we show that even without the fixpoint operators
this probabilistic mu-calculus can be used to characterise these
behavioural relations in the sense that two states are equivalent
if and only if they satisfy the same set of formulae.
\end{abstract}

\section{Introduction}

In concurrency theory, behavioural relations such as equivalences and
refinement preorders form a basis for establishing system correctness.
Usually both specifications and implementations are expressed as
processes within the same framework, in which a specification describes
some high-level behaviour and an implementation gives the technical
details for achieving the behaviour. Then one chooses an equivalence
or preorder to verify that the implementation realises the behaviour
required by the specification.

A great many behavioural relations are defined on top of labelled
transition systems, which offer an operational model of systems. For
finitary (i.e.\ finite-state and finitely branching) systems, these
behavioural relations can be computed in a mechanical way, and thus
may be incorporated into automatic verification tools. In recent
years, probabilistic constructs have been proven useful for giving
quantitative specifications of system behaviour. The first papers on
probabilistic concurrency theory \cite{GJS90,Chr90,LS91}
proceed by \emph{replacing} nondeterministic with probabilistic
constructs. The reconciliation of nondeterministic and probabilistic
constructs starts with \cite{HJ90} and has received a lot of
attention in the literature
\cite{WL92,SL94,Lowe95,Seg95,HK97,MM97,BS01,JW02,mislove03,CIN05,TKP05,MM07,DGHMZ07,DGHM08,DWTC09}.
We shall also work in a framework that features the co-existence of
probability and nondeterminism.

Among the behavioural relations that have proven useful in
probabilistic concurrency theory are various types of
\emph{simulation} and \emph{bisimulation} relations.
Axiomatisations for bisimulations have been investigated in
\cite{BS01,DP07}. Logical characterisations of bisimulations and
simulations have been studied in \cite{SL94,PS07}. For example, in
\cite{SL94} the probabilistic computation tree logic (PCTL) \cite{HJ94}
is used and it turns out that two states are
bisimilar if and only if they satisfy the same set of PCTL formulae.

In the nonprobabilistic setting, there is a line of research on
characteristic formulae. The goal is to seek a particular formula
$\phi_s$ for a given state $s$ such that a necessary and sufficient
condition for any state $t$ being bisimilar to $s$ is to satisfy
$\phi_s$ \cite{SI94}. This is a very strong property in the sense
that to check if $t$ is bisimilar to $s$ it suffices to consider the
single formula $\phi_s$ and see if it can be satisfied by $t$. It
offers a convenient method for equivalence or preorder checking.

In this paper we partially extend the results of \cite{SI94} to a
probabilistic setting that admits both probabilistic and
nondeterministic choice; to make the main ideas neat we do not
consider divergence. We present a probabilistic extension of the
modal mu-calculus \cite{Koz83} (pMu), where a formula is interpreted
as the set of probability distributions satisfying it. This is in
contrast to the probabilistic semantics of the mu-calculus as
studied in
 \cite{HK97,MM97,MM07} where formulae denote lower bounds of
probabilistic evidence of properties, and the semantics of the
generalised probabilistic logic of \cite{CIN05} where a mu-calculus
formula is interpreted as a set of deterministic trees that satisfy
it.

We shall provide characteristic formulae for strong and weak
probabilistic (bi)simulation as introduced in \cite{SL94,Seg95}, as
well as forward simulation \cite{Seg95} and failure simulation
\cite{DGHM08}. The results are obtained in two phases, which we
illustrate by taking strong probabilistic bisimilarity $\BISI$ as an
example. Given a finite-state probabilistic labelled transition system
with state space $\{s_1,...,s_n\}$,
we first construct an equation system $E$ of modal formulae in pMu.
\[\begin{array}{rcl}
E: X_{s_1} & = & \phi_{s_1} \\
       & \vdots & \\
   X_{s_n} & = & \phi_{s_n}
\end{array}\]
A solution of the equation system is a function $\rho$ that assigns
to each variable $X_{s_i}$ a set of distributions $\rho(X_{s_i})$.
The greatest solution of the equation system, denoted by $\nu_E$,
has the property that $s_i\BISI s_j$ if and only if the point
distribution $\pdist{s_j}$ is an element of $\nu_E(X_{s_i})$. In the
second phase, we apply three transformation rules upon $E$ in order
to obtain a pMu formula $\phi^\sim_{s_i}$ whose  meaning
$\Op{\phi^\sim_{s_i}}$ is exactly captured by $\nu_E(X_{s_i})$.  As
a consequence, we derive a characteristic formula for $s_i$ such
that $s_i\BISI s_j$ if and only if
$\pdist{s_j}\in\Op{\phi^\sim_{s_i}}$.

Without the fixpoint operators pMu gives rise to a probabilistic
extension of the Hennessy-Milner logic \cite{HM85}. In analogy to the
nonprobabilistic setting, it characterises (bi)simulations in the
sense that $s\BISI t$ if and only if the two states $s,t$ satisfy
the same set of formulae.

 The paper is
organised as follows. In Section~\ref{s:plts} we recall the
definitions of several (bi)simulations defined over probabilistic
labelled transition systems. In  Section~\ref{s:pmu} we introduce
the syntax and semantics of pMu. In Section~\ref{s:ces}
we build characteristic equation systems and derive from them
characteristic formulae for all our (bi)simulations. In
Section~\ref{s:mod.char} we consider the fixpoint-free fragment of pMu
which characterises a state by the class of formulae it satisfies.
Finally, in Section~\ref{s:cr} we provide some concluding remarks.

\section{Probabilistic (bi)simulations}\label{s:plts}
In this section we recall several probabilistic extensions of
simulation and bisimulation \cite{Mil89a} that appeared in the
literature.

We begin with some notation concerning probability distributions.
A \emph{(discrete) probability distribution} over a set $S$ is a function
{$\Delta\!: S \rightarrow [0,1]$} with \mbox{$\sum_{s\in S}\!
\Delta(s)=1$}; the \emph{support} of $\Delta$ is given by
$\support{\Delta} = \setof{s \in S}{\Delta(s) > 0}$. We write
$\dist{S}$, ranged over by $\Delta, \Theta$, for the set of all
distributions over $S$.
We also write $\pdist{s}$ to denote the point distribution assigning
probability 1 to $s$ and 0 to all others, so that
$\support{\pdist{s}} = \{s\}$.
If $p_i\geq 0$ and $\Delta_i$ is a distribution for each $i$ in some
index set $I$, and $\sum_{i\in I}p_i = 1$, then the probability
distribution $\sum_{i \in I}p_i \cdot \Delta_i \in \dist{S}$ is
given by $  (\sum_{i \in I}p_i \cdot \Delta_i)(s) = \sum_{i \in I}
p_i \cdot\Delta_i(s)$; we will sometimes write it as $p_1 \cdot
\Delta_1 + \ldots + p_n \cdot \Delta_n$ when $I =
\sset{1,\ldots,n}$.

We now present the operational model that we shall use in the
remainder of the paper.

\begin{definition}\rm
A \emph{finite state probabilistic labelled transition system} (pLTS)
is a triple $\langle S, \Act_\tau, \rightarrow \rangle$, where
\begin{enumerate}\vspace{-1ex}
\item $S$ is a finite set of states
\item $\Act_\tau$ is a set of external actions $\Act$
  augmented with an internal action $\tau\mathbin{\not\in}\Act$
\item $\mathord\rightarrow \;\subseteq\; S \times \Act_\tau \times \dist{S}$.
\end{enumerate}
\end{definition}
We usually write $s \ar{a} \Delta$ for $(s,a,\Delta) \inp
\mathord\rightarrow$, $s \ar{a}$ for $\exists \Delta : s
\ar{a}\Delta$, and $s \nar{a}$ for the negation of $s\ar{a}$. We
write $s\nar{A}$ with $A\subseteq\Act$ when $\forall a\in
A\cup\{\tau\}: s\nar{a}$, and $\Delta\nar{A}$ when $\forall
s\in\support{\Delta}: s\nar{A}$. A pLTS is \emph{finitely branching}
if, for each state $s$, the set $\sset{(a,\Delta)\mid
s\ar{a}\Delta}$ is finite. A pLTS is
\emph{finitary} if it is finite-state and finitely branching.

To define probabilistic (bi)simulations, it is often necessary to lift a
relation over states to one over distributions.
\begin{definition}\label{d:lift}\rm
Given two sets $S$ and $T$ and a relation $\mathord{\CR} \subseteq S
\mathop{\times} T$. We lift $\CR$ to a relation
$\mathord{\lift{\CR}} \subseteq \dist{S} \mathop{\times} \dist{T}$
by letting $\Delta \lift{\CR} \Theta$ whenever
\begin{enumerate}\vspace{-1ex}
\item
$\Delta = \sum_{i\in I}{p_i \cdot \pdist{s_i}}$, where
  $I$ is a countable index set and  $\sum_{i\in I}{p_i} = 1$

\item
for each $i \in I$ there is a state $t_i$ such that $s_i \aRel t_i$

\item
$\Theta = \sum_{i \in I} {p_i \cdot \pdist{t_i}}$.
\end{enumerate}
\end{definition}

\noindent Note that in the decomposition of $\Delta$, the states
$s_i$  are not necessarily distinct: that is, the decomposition is
not in general unique, and similarly for the decomposition of
$\Theta$. For example, if
$\CR=\sset{(s_1,t_1),(s_1,t_2),(s_2,t_3),(s_3,t_3)}$,
$\Delta=\frac{1}{2}\pdist{s_1}+\frac{1}{4}\pdist{s_2}+\frac{1}{4}\pdist{s_3}$,
and
$\Theta=\frac{1}{3}\pdist{t_1}+\frac{1}{6}\pdist{t_2}+\frac{1}{2}\pdist{t_3}$,
then $\Delta\lift{\CR}\Theta$ holds because of the decompositions
$\Delta=\frac{1}{3}\pdist{s_1}+\frac{1}{6}\pdist{s_1}+\frac{1}{4}\pdist{s_2}+\frac{1}{4}\pdist{s_3}$
and
$\Theta=\frac{1}{3}\pdist{t_1}+\frac{1}{6}\pdist{t_2}+\frac{1}{4}\pdist{t_3}+\frac{1}{4}\pdist{t_3}$.

From the above definition, the next two properties follow. In fact,
they are sometimes used as alternative methods of lifting relations
(see e.g. \cite{SL94,LS91}).
\begin{proposition}\rm\label{p:lift.char}
\begin{enumerate}
\item Let $\Delta$ and $\Theta$ be distributions over $S$ and $T$,
  respectively. Then
  $\Delta \lift{\aRel} \Theta$ iff there
  exists a weight function $w:S\times T \rightarrow [0,1]$ such that
  \begin{enumerate}
  \item $\forall s\in S: \sum_{t\in T}w(s,t) = \Delta(s)$
  \item $\forall t\in T: \sum_{s\in S}w(s,t) = \Theta(t)$
  \item $\forall (s,t)\in S\times T: w(s,t) > 0 \Rightarrow s\aRel t$.
  \end{enumerate}
\item Let $\Delta,\Theta$ be distributions over  $S$ and
 $\CR$ be an equivalence relation. Then
$\Delta \lift{\aRel} \Theta$ iff $\Delta(C)=\Theta(C)$ for all
equivalence classes $C\in S/\CR$, where $\Delta(C)$ stands for the
accumulated probability $\sum_{s\in C}\Delta(s)$.
\end{enumerate}
\end{proposition}
\begin{proof}
See Proposition 2.3 in \cite{DD09}.
\hfill\qed
\end{proof}

\noindent In a similar way, following \cite{DGHMZ07}, we can lift a
relation $\CR\subseteq S\times\dist{T}$ to a relation
$\mathord{\lift{\CR}}\subseteq \dist{S}\times\dist{T}$, by letting
$\Delta \lift{\CR} \Theta$ whenever
\begin{enumerate}\vspace{-1ex}
\item
$\Delta = \sum_{i\in I}{p_i \cdot \pdist{s_i}}$, where
  $I$ is a countable index set and  $\sum_{i\in I}{p_i} = 1$

\item
for each $i \in I$ there is a distribution $\Theta_i$ such that $s_i
\aRel \Theta_i$

\item
$\Theta = \sum_{i \in I} {p_i \cdot \Theta_i}$.
\end{enumerate}

The above lifting constructions satisfy the following two useful
properties, whose proofs are easy, so we omit them.
\begin{proposition}\rm\label{prop:lifting}
Suppose $\mathord{\aRel}\subseteq S\times S$ or $S \mathop{\times}
 \dist{S}$ and
 $\sum_{i\in I} p_i = 1$. Then
\begin{enumerate}\vspace{-1ex}
\item
$\Delta_i \lift{\aRel} \Theta_i$ for all $i\inp I$ implies $(\sum_{i\in I}{p_i \cdot
\Delta_i}) \lift{\aRel} (\sum_{i\in I}{p_i \cdot \Theta_i})$.

\item
If $(\sum_{i\in I}{p_i \cdot \Delta_i}) \lift{\aRel} \Theta$ then
$\Theta = \sum_{i\in I}{p_i \cdot \Theta_i}$  for some set of
distributions $\Theta_i$ such that $\Delta_i \lift{\aRel} \Theta_i$ for all $i\inp I$.
\hfill\qed
\end{enumerate}
\end{proposition}

\noindent We write $s \ar{\hat{\tau}} \Delta$ if either $s \ar{\tau}
\Delta$ or $\Delta = \pdist{s}$, and  $s \ar{\hat{a}} \Delta$ iff
$s\ar{a}\Delta$ for $a\in\Act$. For any $a\in\Act_\tau$, we know
that $\mathord{\ar{\hat{a}}}\subseteq S\times\dist{S}$, so we can
lift it to be a transition relation between distributions. With a
slight abuse of notation we simply write $\Delta\ar{\hat{a}}\Theta$
for $\Delta\lift{(\ar{\hat{a}})}\Theta$. Then we define weak
transitions $\darhat{{a}}$ by letting $\darhat{{\tau}}$ be the
reflexive and transitive closure of $\ar{\hat{\tau}}$ and writing
$\Delta\darhat{{a}} \Theta$ for $a \inp\Act$ whenever $\Delta
\darhat{{\tau}} \ar{\hat{a}} \darhat{{\tau}} \Theta$.

\begin{definition}\rm
A \emph{divergence} is a sequence of states $s_i$ and distributions
$\Delta_i$ with $s_i\ar{\tau} \Delta_{i}$ and $s_{i+1}\inp\support{\Delta_i}$
for $i\geq 0$.
\end{definition}
The above definition of $\darhat{ a}$ is sensible only in the absence
of divergence. In general, one would need a more complicated notion of
$\darhat{ a}$, such as proposed in \cite{DGHM09}.
Therefore, from here on we restrict attention to divergence-free pLTSs.

\begin{definition}\rm\label{d:sbisi}
A relation $\mathord{\CR}\subseteq S \times S$ is a {\em strong
probabilistic simulation} if $s\ \CR\ t$ and $a \inp \Act_\tau$
implies
\begin{itemize}\vspace{-1ex}
\item if $s\ar{a}\Delta$ then there exists some $\Theta$ such that
  $\pdist{t}\ar{a}\Theta$ and $\Delta \lift{\aRel} \Theta$
\end{itemize}
If both $\CR$ and $\CR^{-1}$ are strong probabilistic simulations,
then $\CR$ is a {\em strong probabilistic bisimulation}.
A state $s$ is related to another state $t$ via \emph{strong
  probabilistic similarity} (resp.\ \emph{bisimilarity}), denoted
 $s \SI t$ (resp. $s \BISI t$), if there
exists a strong probabilistic simulation (resp. bisimulation) $\CR$
such that $s\
  \CR\ t$. \emph{Weak probabilistic similarity} ($\WSI$) and
  \emph{weak probabilistic bisimilarity} ($\WBISI$) are
  defined in the same manner just by using  $\pdist{t}\darhat{{a}}\Theta$
  in place of $\pdist{t}\ar{a}\Theta$.
\end{definition}

\noindent All four (bi)simulations above stem from
\cite{SL94,Seg95}. There they were proposed as improvements over the
strong bisimulation of \cite{HJ90} and the strong simulation of
\cite{JL91}, both of which can be defined as the strong
\emph{probabilistic} (bi)simulation above, but using  $t\ar{a}\Theta$
  in place of $\pdist{t}\ar{a}\Theta$.
Other definitions of simulation have also
appeared in the literature. Here we consider two typical ones:
forward simulation \cite{Seg95} and failure simulation
\cite{DGHM08}.

\begin{definition}\rm\label{d:failsim}
A relation $\mathord{\aRel}\subseteq S\times\dist{S}$ is a {\em failure
  simulation} if $s\aRel\Theta$ implies
\begin{enumerate}\vspace{-1ex}
\item if $s\ar{a}\Delta$ with $a\inp\Act_\tau$ then $\exists\Theta'$ such that
  $\Theta\darhat{{a}}\Theta'$ and $\Delta\lift{\aRel}\Theta'$;\vspace{2pt}
\item if $s\nar{A}$ with $A\subseteq\Act$ then $\exists\Theta'$ such
  that $\Theta\darhat{{\tau}}\Theta'$ and $\Theta'\nar{A}$.
\vspace{-1ex}\end{enumerate}
We write $s\failsim \Theta$ if there is some failure simulation
$\aRel$ such that $s\aRel \Theta$.
\end{definition}
Similarly, we define a forward
simulation and $s\forsim\Theta$ by dropping the second clause in
Definition~\ref{d:failsim}.

\begin{lemma}\label{transfer property}
Let $\mathord{\aRel} \in \{\WBISI, \WSI, \forsim, \failsim\}$.
\begin{enumerate}\vspace{-1ex}
\item
If $\Delta \lift{\aRel} \Theta$ and $\Delta \ar{a} \Delta'$ then
$\exists\Theta'$ such that $\Theta\darhat{{a}}\Theta'$ and
$\Delta\lift{\aRel}\Theta'$.
\item
If $\Delta \lift{\aRel} \Theta$ and $\Delta \darhat{a} \Delta'$ then
$\exists\Theta'$ such that $\Theta\darhat{{a}}\Theta'$ and
$\Delta\lift{\aRel}\Theta'$.
\vspace{-1ex}\end{enumerate}
If $\mathord{\aRel} \in \{\BISI, \SI\}$, the first result applies as
well, but with $\ar{a}$ instead of $\darhat{a}$.
\end{lemma}

\begin{proof}
We start with the cases that $\mathord{\aRel} \mathbin{=} \mathord{\WBISI}$ or
$\mathord{\aRel} \mathbin{=} \mathord{\WSI}$.
Let $\Delta\lift{\aRel} \Theta$ and \plat{$\Delta\lift{\ar{a}}\Delta'$}.
The latter means that $\Delta = \sum_{i\in I} p_i \cdot \pdist{s_i}$,
$\Delta' = \sum_{i\in I} p_i \cdot \Delta'_i$ and $s_i \ar{a}\Delta'_i$
for $i\in I$.  Since \plat{$\Delta\lift{\aRel} \Theta$}, we have
$\Theta = \sum_{i\in I} p_i \cdot \Theta_i$ with
\plat{$\pdist{s_i}\lift{\aRel} \Theta_i$}, using Proposition~\ref{prop:lifting}(2).
Therefore, for each $i\inp I$ and $t\in\support{\Theta_i}$, we have
$s_i \aRel t$, and hence there is some $\Theta'_t$ with
$\pdist{t}\darhat{{a}}\Theta'_t$ and $\Delta'_i\lift{\aRel} \Theta'_t$.
Let $\Theta'_i := \sum_t\Theta_i(t)\cdot \Theta'_t$.  Then
$\Theta_i\darhat{a}\Theta'_i$ and \plat{$\Delta'_i\lift{\aRel} \Theta'_i$},
using Lemma 6.6 from \cite{DGHMZ07}, which is
Proposition~\ref{prop:lifting}(1) but with $\darhat{a}$ instead of $\lift{\aRel}$.
Let $\Theta' := \sum_{i\in I}p_i \cdot\Theta'_i$.
Then $\Theta\darhat{a}\Theta'$ and \plat{$\Delta'\lift{\aRel} \Theta'$},
again by Lemma 6.6 of \cite{DGHMZ07}.

The first statement, and its proof, also hold with $\ar{\hat{\tau}}$
instead of $\ar{a}$. From this, the second statement follows by transitivity.

The cases that $\mathord{\aRel} = \mathord{\forsim}$ or
$\mathord{\aRel} = \mathord{\failsim}$ proceed likewise, except that
the two sentences starting with ``Therefore'' are replaced by:\\
Therefore,
for each $i\in I$ there are some index set $J_i$ and probabilities $p_{ij}$ such that $\sum_{j\in J_i}p_{ij}=1$ and
 $\Theta_i \,=\, \sum_{j\in J_i}p_{ij}\cdot\Theta_{ij}$ with $s_i \aRel
\Theta_{ij}$ for all $j\inp J_i$, and hence
there are $\Theta'_{ij}$ with \plat{$\Theta_{ij}\darhhat{{a}}\Theta'_{ij}$}
and \plat{$\Delta'_i\lift{\aRel} \Theta'_{ij}$}.
Let $\Theta'_i := \sum_j p_{ij}\cdot \Theta'_{ij}$.

The proof for $\mathord{\aRel} \mathbin{=} \mathord{\BISI}$ or
$\mathord{\aRel} \mathbin{=} \mathord{\SI}$ goes as for
$\mathord{\aRel} \mathbin{=} \mathord{\WBISI}$, with $\ar{a}$
replacing $\darhat{a}$.
\qed
\end{proof}

\section{The Probabilistic Modal mu-Calculus}\label{s:pmu}
Let $\Var$ be a countable set of variables.
We define a class $\cL^{\rm raw}$ of modal formulae
by the following grammar:
\[ \phi := 
\bigwedge_{i\in I}\phi_i \mid \bigvee_{i\in I}\phi_i \mid \neg \phi \mid
\diam{a}\phi \mid \boxm{a}\phi \mid
\bigoplus_{i\in I} \phi_i \mid 
\bigoplus_{i\in I}p_i\cdot\phi_i \mid \,\downarrow\!\!\phi \mid X \mid
\mu X.\phi \mid \nu X.\phi \]
where $I$ is an index set, $a\in\Act_\tau$ and
$\sum_{i\in I}p_i=1$. The \emph{probabilistic modal mu-calculus} (pMu)
is given by the subclass $\cL$, obtained by imposing the syntactic
condition that in $\mu X.\phi$ and $\nu X.\phi$ the variable $X$ may
occur in $\phi$ only within the scope of an even number of negations.
The above syntax is obtained by adding a variant of the probabilistic construct
$\bigoplus_{i\in I}p_i\cdot\phi_i$, introduced in \cite{DGHM08} in the
context of a less expressive logic without fixpoint operators, as well
as the novel modalities $\bigoplus_{i\in I}\phi_i$ and $\downarrow\!\!\phi$,
to the syntax of the non-probabilistic mu-calculus \cite{Koz83}.  As
usual, one has $\bigwedge_{i\in\emptyset}\phi_i=\true$
and $\bigvee_{i\in\emptyset}\phi_i=\false$.

The two fixpoint operators $\mu X$ and $\nu X$ bind the respective
variable $X$. We apply the usual terminology of free and bound
variables in a formula and write $\fv(\phi)$ for the set of free
variables in $\phi$. A formula $\phi$ is \emph{closed} if
$\fv(\phi)=\emptyset$.

For any set $\Omega$, write $\pow{\Omega}$ for the power set of $\Omega$.
We use {\em environments}, which bind free variables to
sets of distributions, in order to give semantics to formulae. Let
\[\Env=\setof{\rho}{\rho:\Var\rightarrow\pow{\dist{S}}}\] be the set of
all environments and ranged over by $\rho$. For a set
$V\subseteq\dist{S}$ and a variable $X\in\Var$, we write
$\rho[X\mapsto V]$ for the environment that maps $X$ to $V$ and $Y$ to
$\rho(Y)$ for all $Y\not=X$.

The semantics of a formula $\phi$ in an environment $\rho$ is given as
the set of distributions $\Op{\phi}_\rho$ satisfying it. This leads to
a semantic functional $\Op{\ }:\cL \rightarrow \Env \rightarrow \pow{\dist{S}}$
defined inductively in Table \ref{f:semantics}.
As the meaning of a closed formula $\phi$ does not depend on the
environment, one writes $\Op{\phi}$ for $\Op{\phi}_\rho$ where $\rho$
is an arbitrary environment.
In that case one also writes $\Delta\models\phi$ for $\Delta\in\Op{\phi}$.

\begin{table}
\begin{center}
$\begin{array}{|rcl@{\qquad\qquad\mbox{so}\qquad\qquad}rcl|}
\hline
&&\multicolumn{4}{l|}{~}\\[-8pt]
\Op{\bigwedge_{i\in I}\phi_i}_\rho & = &
\bigcap_{i\in I}\Op{\phi_i}_\rho &
\Op{\true}_\rho & = & \dist{S} \\
\Op{\bigvee_{i\in I}\phi_i}_\rho & = & \bigcup_{i\in
  I}\Op{\phi_i}_\rho &
\Op{\false}_\rho & = & \emptyset \\
\Op{\neg \phi}_\rho & = & \multicolumn{4}{l|}{\dist{S}\setminus\Op{\phi}_\rho} \\
\Op{\diam{a}\phi}_\rho & = & \multicolumn{4}{l|}{ \setof{\Delta\in\dist{S}}{\exists
  \Delta':\Delta \ar{a} \Delta'\ \wedge\ \Delta'\in\Op{\phi}_\rho} }\\
\Op{\boxm{a}\phi}_\rho & = &  \multicolumn{4}{l|}{\setof{\Delta\in\dist{S}}{\forall
  \Delta':\Delta \ar{a} \Delta'\ \Rightarrow\ \Delta'\in\Op{\phi}_\rho}}\\
\Op{\bigoplus_{i\in I}\phi_i}_\rho & = & \multicolumn{4}{l|}{
\{\, \Delta\in\dist{S} \,\mid\, \begin{array}[t]{@{}l@{}}
    \Delta= \sum_{i\in I}p_i\cdot\Delta_i \mbox{ for some $p_i$ with $\sum_{I\in I}p_i=1$}\\
    \wedge\ \forall i\inp I:~ \Delta_i\in\Op{\phi_i}_\rho \,\}\end{array}} \\
\Op{\bigoplus_{i\in I}p_i\cdot\phi_i}_\rho & = & \multicolumn{4}{l|}{
\setof{\Delta\in\dist{S}}{\Delta= \sum_{i\in I}p_i\cdot\Delta_i
  \ \wedge\ \forall i\in I:~ \Delta_i\in\Op{\phi_i}_\rho}} \\
\Op{\downarrow\!\!\phi}_\rho & = & \multicolumn{4}{l|}{
\setof{\Delta\in\dist{S}}{\forall s\in\support{\Delta}: \pdist{s}\in\Op{\phi}_\rho}} \\
\Op{X}_\rho & = & \multicolumn{4}{l|}{ \rho(X) } \\
\Op{\mu X.\phi}_\rho & = & \multicolumn{4}{l|}{
\bigcap\setof{V\subseteq\dist{S}}{\Op{\phi}_{\rho[X\mapsto V]} \subseteq V}}\\
\Op{\nu X.\phi}_\rho & = & \multicolumn{4}{l|}{
\bigcup\setof{V\subseteq\dist{S}}{\Op{\phi}_{\rho[X\mapsto V]} \supseteq V}}\\[2pt]
\hline
&&\multicolumn{4}{l|}{~}\\[-8pt]
\Op{\diam{a}\phi}_\rho & = & \multicolumn{4}{l|}{\setof{\Delta\in\dist{S}}{\exists
  \Delta':\Delta \darhat{{a}} \Delta'\ \wedge\ \Delta'\in\Op{\phi}_\rho}}\\
\Op{\boxm{a}\phi}_\rho & = & \multicolumn{4}{l|}{\setof{\Delta\in\dist{S}}{\forall
  \Delta':\Delta \darhat{{a}} \Delta'\ \Rightarrow\ \Delta'\in\Op{\phi}_\rho}}\\
\hline
\end{array}$
\end{center}
\vspace{-1ex}
\caption{Strong and weak semantics of the probabilistic modal mu-calculus}\label{f:semantics}
\vspace{-1ex}
\end{table}

Following \cite{Koz83,RS97} we give a \emph{strong} and a \emph{weak}
semantics of the probabilistic modal mu-calculus. Both are
the same as those of the modal mu-calculus \cite{Koz83,RS97} except
that distributions of states are taking the roles of states. The
power set of $\dist{S}$, $\pow{\dist{S}}$, may be viewed as the
complete lattice $(\pow{\dist{S}},\linebreak[1]\dist{S},\emptyset,
\subseteq,\cup,\cap)$. Intuitively, we identify a formula with the
set of distributions that make it true. For example, $\true$ holds
for all distributions and dually $\false$ holds for no distribution.
Conjunction and disjunction are interpreted by intersection and
union of sets, and negation by complement. The formula $\diam{a}\phi$
holds for a distribution $\Delta$ if there is a distribution $\Delta'$
that can be reached after an $a$-transition and that satisfies
$\phi$. Dually, $\boxm{a}\phi$ holds for $\Delta$ if all distributions
reachable from $\Delta$ by an $a$-transition satisfy $\phi$. The
formulas $\bigoplus_{i\in I}\phi_i$ and $\bigoplus_{i\in
  I}p_i\cdot\phi_i$ hold for $\Delta$ if the distribution can be
decomposed into a convex combination of some distributions $\Delta_i$
and each of them satisfies the corresponding sub-formula $\phi_i$; the
first of these modalities allows \emph{any} convex combination,
whereas the second one specifies a particular one. The formula
$\downarrow\!\!\phi$ holds for $\Delta$ if all states in its support satisfy $\phi$.
The characterisation of the {\em least fixpoint formula} $\mu X.\phi$
and the {\em greatest fixpoint formula} $\nu X.\phi$ follows from the
well-known Knaster-Tarski fixpoint theorem \cite{Tar55}.

The weak semantics reflects the unobservable nature of internal
actions; it differs from the strong semantics only in the use of the
relations $\darhat{{a}}$ instead of $\ar{a}$ in the interpretation of
the modalities $\diam{a}$ and $\boxm{a}$.

Note that there is some redundancy in the syntax of pMu: each of the
constructs $\bigwedge_{i\in I}$, $\diam{a}$ and $\mu$ can be expressed
in terms of its dual $\bigvee_{i\in I}$, $\boxm{a}$ and $\nu$ with the
aid of negation. However, negation may not be redundant, as the dual
of $\bigoplus_{i\in I}p_i\cdot\phi_i$ does not appear to be expressible without
using negation; moreover this dual lacks the intuitive appeal for
introducing it as a new primitive.

We shall consider (closed) {\em equation systems} of formulae of the
form
\[\begin{array}{rcl}
E: X_1 & = & \phi_1 \\
       & \vdots & \\
   X_n & = & \phi_n
\end{array}\]
where $X_1,...,X_n$ are mutually distinct variables and
$\phi_1,...,\phi_n$ are formulae having at most $X_1,...,X_n$ as
free variables. Moreover, each occurrence of

Here $E$ can be viewed as a function
$E:\Var\rightarrow\cL$ defined by $E(X_i)=\phi_i$ for $i=1,...,n$
and $E(Y)=Y$ for other variables $Y\in\Var$.

An environment $\rho$ is a {\em solution} of an equation system $E$ if
its assignment to $X_i$ coincides with the interpretation of $\phi_i$
in the environment, that is,
\[\forall i:\rho(X_i) = \Op{\phi_i}_\rho.\]
The existence of solutions
for an equation system can be seen from the following arguments. The
set $\Env$, which includes all candidates for solutions, together with
the partial order $\sqsubseteq$ defined by
\[\rho\sqsubseteq\rho'\ \mbox{\rm  iff}\ \forall
X\in\Var:\rho(X)\subseteq\rho'(X)\] forms a complete lattice. The
{\em equation functional} $\cF_E:\Env\rightarrow \Env$ given in the
notation of the $\lambda$-calculus by \[\cF_E:=\lambda\rho.\lambda
X.\Op{E(X)}_\rho\] is monotonic, which can be shown by induction on
the structure of $E(X)$. Thus, the Knaster-Tarski fixpoint theorem
guarantees existence of solutions, and the greatest solution
\begin{equation}\label{e:gfp}
\nu_E:= \bigsqcup \setof{\rho}{\rho\sqsubseteq\cF_E(\rho)}
\end{equation}
is the supremum of the set of all post-fixpoints of $\cF_E$.

An expression $\nu_E(X)$, with $X$ one of the variables used in $E$,
denotes a set of distributions. Below we will use such expressions
as if they were valid syntax in our probabilistic mu-calculus, with
$\Op{\nu_E(X)}_\rho := \nu_E(X)$. This amounts to extending the
greatest fixpoint operator $\nu$ to apply to finite sets of fixpoint
equations, instead of single equations; the expression $\nu X.\phi$
amounts to the special case $\nu_E(X)$ in which $E$ consists of the
single equation $X = \phi$.

\begin{table}[t]
\begin{itemize}
\item Rule 1: $E \rightarrow F$
\item Rule 2: $E \rightarrow G$
\item Rule 3: $E \rightarrow H$ if $X_n\not\in\fv(\phi_1,...,\phi_n)$
\end{itemize}
\framebox[12.2cm][c]{ $\begin{array}{@{}r@{}rclr@{}rclr@{}rclr@{}rcl@{}}
 E:& X_1 & = & \phi_1 &
 F:& X_1 & = & \phi_1 \qquad & G:& X_1 & = &
       \phi_1[\phi_n/X_n] \qquad & H:& X_1 & = & \phi_1 \\
 &  & \vdots &  &
       & & \vdots & & &  & \vdots & && & \vdots & \\
 & X_{n-1} & = & \phi_{n-1} &
    & X_{n-1} & = & \phi_{n-1}& & X_{n-1} & = &
       \phi_{n-1}[\phi_n/X_n] & & X_{n-1} & = & \phi_{n-1} \\
 & X_n & = & \phi_n &
    & X_n & = & \nu X_n.\phi_n & &  X_n & = & \phi_n & &
       &
\end{array}$}
\vspace{1ex}
\caption{Transformation rules}\label{f:rules}
\vspace{-1ex}
\end{table}

The use of expressions $\nu_E(X)$ is justified because they can be
seen as syntactic sugar for authentic pMu expressions. As explained
in \cite{Mul98}, the three transformation rules in
Table~\ref{f:rules} can be used to obtain from an equation system
$E$ a pMu formula whose interpretation coincides with the
interpretation of $X_1$ in the greatest solution of $E$.
\begin{theorem}\label{t:cha.formula}
Given a finite equation system $E$ that uses the variable $X$, there
is a pMu formula $\phi$ such that $\nu_E(X)=\Op{\phi}$.  \hfill\qed
\end{theorem}

\section{Characteristic equation systems}\label{s:ces}
Following \cite{SI94}, the behaviour of a finite-state process can
be characterised by an equation system of modal formulae. In the
current section we show that this idea also applies in the
probabilistic setting.  For each behavioural relation $\CR$ over a
finite state space, ranging over the various simulation preorders
and bisimulation equivalences reviewed in Section~\ref{s:plts}, we
establish an equation system $E$ of modal formulae in pMu.
\[\begin{array}{rcl}
E: X_{s_1} & = & \phi_{s_1} \\
       & \vdots & \\
   X_{s_n} & = & \phi_{s_n}
\end{array}\]
There is exactly one such equation for each state $s_i$, and the
formulae $\phi_{s_i}$ do not contain fixpoint operators.
This equation system is guaranteed to have a
greatest solution $\nu_E$ which has the nice property that, for any
states $s,t$ in the state space in question, $s$ is related to $t$
via $\CR$ if and only if the point distribution $\pdist{t}$ belongs
to the set of distributions assigned to the variable $X_s$ by
$\nu_E$. Thus $\nu_E(X_s)$ is a \emph{characteristic formula} for
$s$ w.r.t.\ $\CR$ in the sense that $s\ \CR\ t$ iff $\pdist{t}$
satisfies $\nu_E(X_s)$.

\subsubsection{Strong probabilistic bisimulation}
The key ingredient for the modal characterisation of strong
probabilistic bisimulation is to construct an equation system that
captures  all the transitions of a pLTS. For each state $s$ we build
an equation $X_s=\phi_s$, where $X_s$ is a variable and $\phi_s$ is
of the form $\phi'_s\wedge\phi''_s$ with $\phi'_s$ a formula
describing the actions enabled by $s$ and $\phi''_s$ a formula
describing the consequences of performing these actions.
Intuitively, if state $s$ is related to state $t$ in a bisimulation
game, then $\phi'_s$ expresses the transitions that should be matched
up by $t$ and $\phi''_s$ expresses the capability of $s$ to match up
the transitions initiated by $t$.
More specifically, the equation system is given by the following
definition.
\begin{definition}\rm\label{d:cess}
Given a pLTS, its {\em characteristic equation system}
for strong probabilistic bisimulation consists of one equation
$X_s=\phi_{s}$ for each state $s\inp S$ where
\begin{equation}\label{e:cf}
\phi_s:=(\bigwedge_{s\stackrel{a}{\longrightarrow}\Delta}\diam{a}X_{\Delta})\wedge
(\bigwedge_{a\in\Act_\tau}\boxm{a} \bigoplus_{s\arh{a}\Delta} X_{\Delta})^{~1}
\end{equation}
with $X_\Delta :=
\bigoplus_{s\in\support{\Delta}}\Delta(s)\cdot {\downarrow}X_s$.
\end{definition}
\footnotetext[1]{The subformula $\bigoplus_{s\arh{a}\Delta} X_{\Delta}$ is
  equivalent to $\bigvee_{\pdist{s}\stackrel{a}{\longrightarrow}\Delta} X_{\Delta}$,
  and this is the form that we use to prove Theorem~\ref{t:ces}.
  If the given pLTS has nondeterministic choices among
  different transitions labelled with the same action, this
  disjunction is infinite. For example, if
  $s\ar{a}\pdist{s_i}$ for $i=1,2$, then
$\pdist{s}\ar{a}\Delta_p$, where $\Delta_p=p\cdot\pdist{s_1}+(1\mathord-p)\cdot\pdist{s_2}$, for any $p\in [0,1]$. 
The set $\sset{\Delta_p \mid p\in [0,1]}$ is uncountable, though it is finitely generable, as the convex closure of the two-element set $\sset{\Delta_0,\Delta_1}$.
The formula $\bigoplus_{s\arh{a}\Delta} X_{\Delta}$ exploits that
  fact to bypass the infinite disjunction; this formula is finite if
  the underlying pLTS is finitary.
}

\noindent The equation system thus constructed, interpreted according
to the strong semantics of pMu, has the required
property, as stated by the theorem below.%
\begin{theorem}\label{t:ces}
Let $E$ be the characteristic equation system for strong
probabilistic bisimulation on a given
pLTS. Then, for all states $s$ and $t$,\vspace{-1ex}
\begin{enumerate}
\item $s\aRel t$ for some strong probabilistic bisimulation $\aRel$ if and only if
$\pdist{t}\in\rho(X_s)$ for some post-fixpoint $\rho$ of $\cF_E$.
\item In particular,
$s\BISI t$ if and only if $\pdist{t}\in \Op{\nu_E(X_s)}$, i.e.,
$\nu_E(X_s)$ is a characteristic formula for $s$
w.r.t.\ strong probabilistic bisimilarity.
\end{enumerate}
\end{theorem}
\begin{proof}
Let $E$ be the characteristic equation system for strong
probabilistic bisimulation on a given pLTS. We only consider the
first statement, from which the second statement follow immediately.

($\Leftarrow$) For this direction, assuming a
  post-fixpoint $\rho$ of $\cF_E$, we construct a probabilistic
  bisimulation relation that includes all state pairs $(s,t)$ satisfying
  $\pdist{t}\in \rho(X_s)$.
Let $\aRel=\setof{(s,t)}{\pdist{t}\in \rho(X_s)}$. We
  first show that
\begin{equation}\label{e:left}
\Theta\in\Op{X_\Delta}_{\rho}\ {\rm implies}\ \Delta\lift{\aRel}\Theta.
\end{equation}
Let $X_\Delta = \bigoplus_{i\in I}p_i\cdot {\downarrow}X_{s_i}$, so that
$\Delta=\sum_{i\in I}p_i\cdot \pdist{s_i}$. Suppose
$\Theta\in\Op{X_\Delta}_{\rho}$. We have that $\Theta=\sum_{i\in I}
p_i\cdot\Theta_i$ and, for all $i\inp I$ and all
$t\inp\support{\Theta_i}$, that $\pdist{t}\in \Op{X_{s_i}}_{\rho}$,
i.e.\ $s_i \aRel t$. It follows that \plat{$\pdist{s_i}
\lift{\aRel}\Theta_i$} and thus \plat{$\Delta \lift{\aRel} \Theta$}, using
Proposition~\ref{prop:lifting}(1).

Now we show that $\aRel$ is a probabilistic bisimulation.
\begin{enumerate}
\item Suppose $s\aRel t$ and $s\ar{a}\Delta$. Then
  $\pdist{t}\in\rho(X_s)\subseteq\Op{\phi_s}_{\rho}$. It follows from
  (\ref{e:cf}) that $\pdist{t}\in\Op{\diam{a}X_\Delta}_{\rho}$. So
  there exists some $\Theta$ such that $\pdist{t}\ar{a}\Theta$ and
  $\Theta \in \Op{X_\Delta}_{\rho}$. Now we apply (\ref{e:left}).

\item Suppose $s\aRel t$ and $t\ar{a}\Theta$. Then
  $\pdist{t}\in\rho(X_s)\subseteq\Op{\phi_s}_{\rho}$. It follows from
  (\ref{e:cf}) that $\pdist{t}\in\Op{\boxm{a}\bigvee_{\pdist{s}
  \stackrel{a}{\longrightarrow}\Delta}
  X_{\Delta}}$. Notice that it
  must be the case that $\pdist{s}\ar{a}$, otherwise, $\pdist{t}\in
  \Op{[a]\false}_{\rho}$ and thus $t\nar{a}$, in contradiction
  with the assumption $t\ar{a}\Theta$. Therefore,
  $\Theta\in\Op{\bigvee_{\pdist{s}\ar{a}\Delta}X_{\Delta}}_{\rho}$,
  which implies
  $\Theta\in\Op{X_\Delta}_{\rho}$ for some $\Delta$ with
  $\pdist{s}\ar{a}\Delta$.
  Now we apply (\ref{e:left}).
\end{enumerate}

($\Rightarrow$)
Given a strong probabilistic bisimulation $\CR$, we construct a
post-fixpoint of $\cF_E$ such that whenever $s \aRel t$ then
$\pdist{t}$ falls into the set of distributions assigned to $X_s$ by
that post-fixpoint.  We define the environment $\rho_\CR$ by
 \[\rho_\CR(X_s):=\setof{\pdist{t}}{s\aRel t}\]
 and show that $\rho_\CR$ is a post-fixpoint of
 $\cF_E$, i.e.
\begin{equation}\label{e:post}
\rho_\CR \sqsubseteq \cF_E(\rho_\CR).
\end{equation}

We first show that
\begin{equation}\label{e:right}
\Delta\lift{\CR}\Theta\ {\rm implies}\
\Theta\in\Op{X_\Delta}_{\rho_\CR}.
\end{equation}
Suppose $\Delta\lift{\CR}\Theta$, we have that (i)
$\Delta=\sum_{i\in I}p_i\cdot\pdist{s_i}$, (ii)
$\Theta=\sum_{i\in I}p_i\cdot\pdist{t_i}$, (iii)
$s_i\aRel t_i$ for all $i\in I$. We know from (iii) that
$\pdist{t_i}\in\Op{X_{s_i}}_{\rho_\CR}$ and thus
$\pdist{t_i}\in\Op{{\downarrow}X_{s_i}}_{\rho_\CR}$. Using (ii) we have that
$\Theta\in\Op{\bigoplus_{i\in I}p_i\cdot {\downarrow}X_{s_i}}_{\rho_\CR}$. Using
(i) we obtain $\Theta\in\Op{X_\Delta}_{\rho_\CR}$.

Now we are in a position to show (\ref{e:post}). Suppose $\pdist{t}\in
\rho_\CR(X_s)$. We must prove that
$\pdist{t}\in \Op{\phi_s}_{\rho_\CR}$, i.e.
\[\pdist{t}\in
(\bigcap_{s\stackrel{a}{\longrightarrow}\Delta}\Op{\diam{a}X_\Delta}_{\rho_\CR})
\cap
(\bigcap_{a\in\Act_\tau}\Op{\boxm{a}\bigvee_{\pdist{s}\stackrel{a}{\longrightarrow}\Delta}X_{\Delta}}_{\rho_\CR})\]
by (\ref{e:cf}). This can be done by showing that $\pdist{t}$
belongs to each of the two parts of the outermost intersection.
\begin{enumerate}
\item
  Assume that $s\ar{a}\Delta$ for some $a\in\Act_\tau$ and
  $\Delta\in\dist{S}$. Since $s\aRel
  t$, there exists some $\Theta$ such that $\pdist{t}\ar{a}\Theta$ and
  $\Delta \lift{\CR} \Theta$. By (\ref{e:right}), we get
  $\Theta\in\Op{X_\Delta}_{\rho_\CR}$. It follows that $\pdist{t}\in
  \Op{\diam{a}X_\Delta}_{\rho_\CR}$.
\vspace{2pt}
\item Let $a \inp Act_\tau$.
 Whenever $\pdist{t}\ar{a}\Theta$, then by $s\aRel t$ there must
 be some $\Delta$ such that $\pdist{s}\ar{a}\Delta$ and $\Delta \lift{\CR}
 \Theta$. By (\ref{e:right}), we get $\Theta\in\Op{X_\Delta}_{\rho_\CR}$
 and thus $\Theta \in \Op{\bigvee_{\pdist{s}\ar{a}\Delta}X_{\Delta}}_{\rho_\CR}$.
 As a consequence,
 $\pdist{t}\in \Op{\boxm{a}\bigvee_{\pdist{s}\ar{a}\Delta}X_{\Delta}}_{\rho_\CR}$.
\hfill\qed
\end{enumerate}
\end{proof}

\subsubsection{Strong probabilistic simulation}
In a simulation game, if state $s$ is related to state $t$, we only
need to check that all transitions initiated by $s$ should be
matched up by transitions from $t$, and we do not care about the
inverse direction: the capability of $s$ to simulate $t$. Therefore,
it is not surprising that characteristic equation systems for strong
probabilistic simulation are defined as in Definition~\ref{d:cess}
except that we drop the second part of the conjunction in
(\ref{e:cf}), so $\phi_s$ takes the form
\begin{equation}\label{e:cfsi}
\phi_s:=\bigwedge_{s\arh{a}\Delta}\diam{a}X_{\Delta}
\end{equation}
With this modification, we have the expected property for strong
probabilistic simulation, which can be shown by using the ideas in
the proof of Theorem~\ref{t:ces}, but with fewer cases to analyse.

\subsubsection{Weak probabilistic bisimulation}
Characteristic equation systems for weak probabilistic bisimulation
are defined as in Definition~\ref{d:cess} except that the weak
semantics of pMu is employed and $\phi_s$ takes the form
\begin{equation}\label{e:cfw}
\phi_s:=(\bigwedge_{s\arh{a}\Delta}\diam{a}X_{\Delta})\wedge
(\bigwedge_{a\in\Act_\tau}\boxm{a}\bigvee_{\pdist{s}\darhhat{a}\Delta}
X_{\Delta} )^{~2}
\end{equation}
\footnotetext[2]{Using results from Markov Decision Processes
  \cite{Put94}, in a finitary pLTS also this infinite disjunction can
  be expressed as finite convex combination; however, we will not
  elaborate this here.}

\noindent
With the above modifications, we have the counterpart of
Theorem~\ref{t:ces}, with a similar proof.

\subsubsection{Weak probabilistic simulation}
Characteristic equation systems for weak probabilistic simulation
are in exactly the same form as characteristic equation systems for
strong probabilistic simulation (cf. (\ref{e:cfsi})),
but using the weak semantics of pMu.

\subsubsection{Forward simulation}
Characteristic equation systems for forward simulation are
in the same form as characteristic equation systems for weak
probabilistic simulation, but with $X_\Delta :=
\bigoplus_{s\in\support{\Delta}}\Delta(s)\cdot X_s$, i.e.\ dropping
the $\downarrow$.

\subsubsection{Failure simulation}
To give a modal characterisations for failure simulation
we need to add modal formulae of the form $\refuse{A}$ with
$A\subseteq\Act$, first introduced in \cite{DGHM08}, to
pMu, with the meaning given by
$$\Op{\refuse{A}}_\rho  =  \setof{\Delta\in\dist{S}}{\exists \Delta':
    \Delta\darhat{\tau}\Delta'\ \wedge\ \Delta'\nar{A}}$$
The formula $\refuse{A}$ holds for $\Delta$ if
by doing internal actions only $\Delta$ can evolve into a distribution
such that no state in its support can perform an action from $A\cup\{\tau\}$.
This time $\phi_s$ takes the form
\begin{equation}\label{e:cf.failsim}
\phi_s:=\left\{\begin{array}{ll}
\bigwedge_{s\ar{a}\Delta}\diam{a}X_{\Delta} & {\rm if}\ s\ar{\tau}\\
(\bigwedge_{s\ar{a}\Delta}\diam{a}X_{\Delta})\wedge
\refuse{\setof{a}{s\nar{a}}} \qquad & {\rm otherwise}
\end{array}\right.
\end{equation}
with $X_\Delta := \bigoplus_{s\in\support{\Delta}}\Delta(s)\cdot X_s$.
Inspired by \cite{DGHM08}, here we distinguish two cases, depending
on the possibility of making
an internal transition from $s$.

\medskip\noindent
In summary, we have the following property.
\begin{theorem}
Let $E_{\SI}$ be the characteristic equation system for strong probabilistic simulation on a given pLTS. Let $E_{\WBISI}$ ($E_{\WSI}, E_{\forsim}, E_{\failsim}$, respectively) be the characteristic equation system for weak probabilistic bisimulation (weak probabilistic simulation, forward simulation, failure
simulation, respectively) on a given divergence-free pLTS. Then, for all
states $s$, $t$ and distributions $\Theta$,\vspace{-1ex}
\begin{enumerate}
\item $s\aRel t$ for some strong probabilistic simulation (weak probabilistic bisimulation, weak probabilistic simulation, respectively) $\aRel$ if and only if
$\pdist{t}\in\rho(X_s)$ for some post-fixpoint $\rho$ of $\cF_{E_{\SI}}$ ($\cF_{E_{\WBISI}}$, $\cF_{E_{\WSI}}$, respectively). 

\item $s\aRel \Theta$ for some forward simulation (failure simulation) $\aRel$ if and only if
$\Theta\in\rho(X_s)$ for some post-fixpoint $\rho$ of $\cF_{E_{\forsim}}$ ($\cF_{E_{\failsim}}$).

\item In particular,
\begin{enumerate}
\item $s\SI t$ if and only if $\pdist{t}\in \Op{\nu_{E_{\SI}}(X_s)}$.
\item $s\WBISI t$ if and only if $\pdist{t}\in \Op{\nu_{E_{\WBISI}}(X_s)}$.
\item $s\WSI t$ if and only if $\pdist{t}\in \Op{\nu_{E_{\WSI}}(X_s)}$.
\item $s\failsim \Theta$ if and only if $\Theta\in \Op{\nu_{E_{\forsim}}(X_s)}$.
\item $s\failsim \Theta$ if and only if $\Theta\in \Op{\nu_{E_{\failsim}}(X_s)}$. \hfill\qed
\end{enumerate}
\end{enumerate}
\end{theorem}

\noindent
We can also consider the strong case for $\forsim$ and $\failsim$ by
treating $\tau$ as an external action, and give characteristic
equation systems. In the strong case for $\failsim$ only the
``otherwise'' in (\ref{e:cf.failsim}) applies, with $\refuse{A}$
represented as $\bigwedge_{a\in A}[a]\false$.

\section{Modal characterisations}\label{s:mod.char}
In the previous sections we have pursued logical characterisations for
various behavioural relations by characteristic formulae. A weaker
form of characterisation, which is commonly called a modal
characterisation of a behavioural relation, consists of isolating a
class of formulae with the property that two states are equivalent if
and only if they satisfy the same formulae from that class.

\begin{definition}\rm\label{d:phml}
Let $\cL^\mu_\BISI$ be simply the class $\cL$ of modal formulae
defined in Section~\ref{s:pmu}, equipped with the strong semantics
of Table~\ref{f:semantics}.
With $\cL^\mu_\SI$ we denote the fragment of this class obtained by
skipping the modalities $\neg$ and $\boxm{a}$.
The classes $\cL^\mu_\WBISI$ and \plat{$\cL^\mu_\WSI$} are defined likewise,
but equipped with the weak semantics.
Moreover, $\cL^\mu_{\forsim}$ is the fragment of \plat{$\cL^\mu_\WSI$}
obtained by skipping $\downarrow$, and
\plat{$\cL^\mu_{\failsim}$} is obtained from \plat{$\cL^\mu_{\forsim}$} by
addition of the modality $\refuse{A}$.

In all cases, dropping the superscript $\mu$ denotes the subclass
obtained by dropping the variables and fixpoint operators.

For $\mathord{\aRel} \inp \{\BISI, \SI, \WBISI, \WSI, \forsim,
\failsim\}$ we write $\Delta \sqsubseteq^\mu_{\aRel} \Theta$ just when
$\Delta\inp\Op{\phi} \mathbin\Rightarrow \Theta\inp\Op{\phi}$ for all
closed $\phi\in\cL^\mu_{\aRel}$, and $\Delta \sqsubseteq_{\aRel} \Theta$ just when
$\Delta\inp\Op{\phi} \Rightarrow \Theta\inp\Op{\phi}$ for all $\phi\in\cL_{\aRel}$.
\end{definition}

\noindent Note that the relations $\sqsubseteq^\mu_\BISI$,
$\sqsubseteq^\mu_\WBISI$, $\sqsubseteq_\BISI$ and
$\sqsubseteq_\WBISI$ are symmetric. For this reason we will employ
the symbol $\equiv$ instead of $\sqsubseteq$ when referring to them.

We have the following modal characterisation for strong
probabilistic bisimilarity, strong probabilistic similarity, weak
probabilistic bisimilarity, weak probabilistic similarity, forward
similarity, and failure similarity.

\begin{theorem}[Modal characterisation]\rm\label{t:modal.characterisation}\\
Let $s$ and $t$ be states in a divergence-free pLTS.\\
\begin{tabular}{r@{.~~}c@{~~iff~~}c@{~~iff~~}c}
1&$s\BISI t$ & $\pdist{s}\equiv^\mu_{\BISI}\pdist{t}$ & $\pdist{s}\equiv_{\BISI}\pdist{t}$.\\
2&$s\SI t$ & $\pdist{s} \sqsubseteq^\mu_\SI \pdist{t}$ & $\pdist{s} \sqsubseteq_\SI \pdist{t}$.\\
3&$s\WBISI t$ & $\pdist{s} \equiv^\mu_\WBISI \pdist{t}$ & $\pdist{s} \equiv_\WBISI \pdist{t}$.\\
4&$s\WSI t$ & $\pdist{s} \sqsubseteq^\mu_\WSI \pdist{t}$ & $\pdist{s} \sqsubseteq_\WSI \pdist{t}$.\\
5&$s\forsim \Theta$ & $\pdist{s} \sqsubseteq^\mu_{\forsim} \Theta$ & $\pdist{s} \sqsubseteq_{\forsim} \Theta$.\\
6&$s\failsim \Theta$ & $\pdist{s} \sqsubseteq^\mu_{\failsim} \Theta$ & $\pdist{s} \sqsubseteq_{\failsim} \Theta$.
\end{tabular}
\end{theorem}
Note that $\pdist{s}\equiv^\mu_{\BISI}\pdist{t} \Rightarrow s\BISI t$
is an immediate consequence of Theorem~\ref{t:ces}:
From $s \BISI s$ we obtain $\pdist{s}\in \Op{\nu_E(X_s)}$.
Together with  $\pdist{s}\equiv^\mu_{\BISI}\pdist{t}$ this yields
$\pdist{t}\in \Op{\nu_E(X_s)}$, hence $s \BISI t$.

\begin{proof}
We only prove the first statement; the others can be shown analogously.
In fact we establish the more general result that
$$\Delta\lift{\BISI} \Theta \qquad\Leftrightarrow\qquad
\Delta\equiv^\mu_{\BISI}\Theta \qquad\Leftrightarrow\qquad
\Delta\equiv_{\BISI}\Theta$$ from which statement 1 of
Theorem~\ref{t:modal.characterisation} follows immediately.  The
implication\linebreak $\Delta\lift{\BISI} \Theta \Rightarrow
\Delta\equiv^\mu_{\BISI}\Theta$ expresses the \emph{soundness} of
the logic $\cL^\mu_\BISI$ w.r.t.\ the relation $\lift{\BISI}$,
whereas the implication $\Delta\equiv_{\BISI}\Theta \Rightarrow
\Delta\lift{\BISI} \Theta$ expresses the \emph{completeness} of
$\cL_\BISI$ w.r.t.\ $\lift{\BISI}$.  The implication
$\Delta\equiv^\mu_{\BISI}\Theta \Rightarrow
\Delta\equiv_{\BISI}\Theta$ is trivial.

(Soundness) An environment $\rho:\Var\rightarrow\pow{\dist{S}}$ is
called \emph{compatible with $\lift{\BISI}$} if for all $X\in\Var$ we have
that $$\Delta\lift{\BISI} \Theta \Rightarrow (\Delta\in\rho(X)
\Rightarrow \Theta\in\rho(X)).$$
We will show by structural induction on $\phi$ that
$$\Delta\lift{\BISI} \Theta \Rightarrow (\Delta\in\Op{\phi}_\rho \Rightarrow
\Theta\in\Op{\phi}_\rho)$$ for any environment $\rho$ that is
compatible with $\lift{\BISI}$. By restricting attention to closed
$\phi$ this implies the soundness of $\cL^\mu_\BISI$ w.r.t.\ $\lift{\BISI}$.
\begin{itemize}
\item Let $\Delta\lift{\BISI} \Theta$ and $\Delta\in\Op{\diam{a}\phi}_\rho$.
  Then $\Delta\ar{a}\Delta'$ and $\Delta'\in\Op{\phi}_\rho$
  for some $\Delta'$. By Lemma~\ref{transfer property}, there is some
  $\Theta'$ with $\Theta\ar{a}\Theta'$ and $\Delta'\lift{\BISI} \Theta'$.
  By induction we have $\Theta'\in\Op{\phi}_\rho$,
  thus $\Theta\models\diam{a}\phi$.
\item Let $\Delta\lift{\BISI} \Theta$ and
  $\Delta\in\Op{\boxm{a}\phi}_\rho$. Suppose $\Theta\ar{a}\Theta'$.
  By Lemma~\ref{transfer property}, and symmetry, there is a $\Delta'$ with $\Delta\ar{a}\Delta'$
  and $\Delta'\lift{\BISI} \Theta'$. As $\Delta\in\Op{\boxm{a}\phi}_\rho$
  it must be that $\Delta'\in\Op{\phi}_\rho$, and by induction we have
  $\Theta'\in\Op{\phi}_\rho$. Thus $\Theta\in\Op{\boxm{a}\phi}_\rho$.
\item Let $\Delta\lift{\BISI} \Theta$ and $\Delta\in \Op{\bigwedge_{i\in I}\phi_i}_\rho$.
  Then $\Delta\in\Op{\phi_i}_\rho$ for all $i\inp I$. So by induction $\Theta\in\Op{\phi_i}_\rho$,
  and we have $\Theta\in\Op{\bigwedge_{i\in I}\phi_i}_\rho$.
\item The case $\Delta\lift{\BISI} \Theta$ and $\Delta\in \Op{\bigvee_{i\in I}\phi_i}_\rho$ goes likewise.
\item Let $\Delta\lift{\BISI} \Theta$ and $\Delta\in\Op{\neg\phi}$. So
  $\Delta\not\in\Op{\phi}$, and by induction (and the symmetry of
  $\lift{\BISI}$) we have $\Theta\not\in\Op{\phi}$. Thus
  $\Theta\in\Op{\neg\phi}$.
\item Let $\Delta\lift{\BISI} \Theta$ and $\Delta\in\Op{\bigoplus_{i\in I}p_i\cdot\phi_i}_\rho$. So
  $\Delta=\sum_{i\in i}p_i\cdot\Delta_i$ and for all $i\in I$
  we have $\Delta_i\in\Op{\phi_i}_\rho$.
  Since $\Delta\lift{\BISI} \Theta$, by
  Proposition~\ref{prop:lifting}(2) we have $\Theta=\sum_{i\in
  I}p_i\cdot\Theta_i$ and $\Delta_i\lift{\BISI}\Theta_i$.
  So by induction we have $\Theta_i\inp\Op{\phi_i}_\rho$ for all $i\inp I$.
  Therefore, $\Theta\in\Op{\bigoplus_{i\in I}p_i\cdot\phi_i}_\rho$.
  The case $\Delta\in\Op{\bigoplus_{i\in I}\phi_i}_\rho$ goes likewise.
\item Let $\Delta\lift{\BISI} \Theta$ and $\Delta\in\Op{{\downarrow}\phi}_\rho$. So
  for all $s\in\support{\Delta}$ we have $\pdist{s}\in\Op{\phi}_\rho$.
  From $\Delta\lift{\BISI} \Theta$ it follows
  that for each $t\inp\support{\Theta}$ there is an
  $s\inp\support{\Delta}$ with $s\BISI t$, thus
  $\pdist{s}\lift{\BISI}\pdist{t}$.
  So by induction we have $\pdist{t}\inp\Op{\phi}_\rho$ for all $t\inp\support{\Theta}$.
  Therefore, $\Theta\in\Op{{\downarrow}\phi}_\rho$.
\item Let $\Delta\lift{\BISI} \Theta$ and $\Delta\in\Op{X}_\rho = \rho(X)$. Then $\Theta \in
  \Op{X}_\rho$ because $\rho$ is compatible with $\lift{\BISI}$.
\item Suppose $\Delta\lift{\BISI} \Theta$ and $\Theta\not\in\Op{\mu X. \phi}_\rho$.
  Then $\exists V\subseteq \dist{S}$ with $\Theta\not\in V$ and
  $\Op{\phi}_{\rho[X\mapsto V]} \subseteq V$. Let $V':=\{\Delta' \mid
  \forall \Theta'. (\Delta'\lift{\BISI}\Theta' \Rightarrow \Theta' \in V)\}$.
  Then $\Delta\not\in V'$. It remains to show that
  $\Op{\phi}_{\rho[X\mapsto V']} \subseteq V'$, because this implies
  $\Delta\not\in\Op{\mu X. \phi}_\rho$, which has to be shown.

  So let $\Delta' \in \Op{\phi}_{\rho[X\mapsto V']}$.
  Take any $\Theta'$ with $\Delta' \lift{\BISI} \Theta'$.
  By construction of $V'$, the environment $\rho[X\mapsto V']$ is
  compatible with $\lift{\BISI}$. Therefore, the induction hypothesis
  yields $\Theta' \in \Op{\phi}_{\rho[X\mapsto V']}$.
  We have $V'\subseteq V$, and as $\Op{\ }$ is monotonic we obtain
  $\Theta' \in \Op{\phi}_{\rho[X\mapsto V']} \subseteq \Op{\phi}_{\rho[X\mapsto V]} \subseteq V$.
  It follows that $\Delta'\in V'$.
\item Suppose $\Delta\lift{\BISI} \Theta$ and $\Delta\in\Op{\nu X. \phi}_\rho$.
  Then $\exists V\subseteq \dist{S}$ with $\Delta\in V$ and
  $\Op{\phi}_{\rho[X\mapsto V]} \supseteq V$. Let $V':=\{\Theta' \mid
  \exists \Delta' \inp V.~ \Delta'\lift{\BISI}\Theta'\}$.
  Then $\Theta\in V'$. It remains to show that
  $\Op{\phi}_{\rho[X\mapsto V']} \supseteq V'$, because this implies
  $\Theta\in\Op{\nu X. \phi}_\rho$, which has to be shown.

  So let $\Theta' \not\in \Op{\phi}_{\rho[X\mapsto V']}$.
  Take any $\Delta'$ with $\Delta' \lift{\BISI} \Theta'$.
  By construction of $V'$, the environment $\rho[X\mapsto V']$ is
  compatible with $\lift{\BISI}$. Therefore, the induction hypothesis
  yields $\Delta' \not\in \Op{\phi}_{\rho[X\mapsto V']}$.
  We have $V'\supseteq V$, and as $\Op{\ }$ is monotonic we obtain
  $\Delta' \not\in \Op{\phi}_{\rho[X\mapsto V']} \supseteq \Op{\phi}_{\rho[X\mapsto V]} \supseteq V$.
  It follows that $\Theta'\not\in V'$.
\end{itemize}

(Completeness) Let $\CR=\sset{(s,t)\mid \pdist{s}\equiv_{\BISI}
\pdist{t}}$. We show that $\CR$ is a strong probabilistic
bisimulation. Suppose $s\ \CR\ t$ and $s\ar{a}\Delta$. We have to
show that there is some $\Theta$ with $\pdist{t}\ar{a}\Theta$ and
$\Delta \lift{\CR} \Theta$. Consider the set
\[T:=\{\Theta \mid \pdist{t}\ar{a}\Theta \wedge \Theta=\!\!\!\sum_{s'\in\support{\Delta}}\Delta(s')\cdot \Theta
_{s'}\wedge \exists s'\in\support{\Delta},\exists
t'\in\support{\Theta_{s'}}:\pdist{s'}\not\equiv_{\BISI}\pdist{t'}\}\]
 For each $\Theta\in T$ there must be some
 $s'_\Theta\in\support{\Delta}$ and
 $t'_\Theta\in\support{\Theta_{s'_\Theta}}$
and a formula $\phi_{\Theta}$ with
$\pdist{s'_\Theta}\models\phi_{\Theta}$ but
$\pdist{t'_\Theta}\not\models\phi_{\Theta}$.
So $\pdist{s'}\models\bigwedge_{\sset{\Theta\in T\mid s'_\Theta =
s'}}\phi_\Theta$ for each $s'\in\support{\Delta}$,
and for each $\Theta\in T$ with $s'_\Theta=s'$
there is some $t'_{\Theta}\in\support{\Theta_{s'}}$ with
$\pdist{t'_{\Theta}}\not\models \bigwedge_{\sset{\Theta\in T\mid
s'_\Theta = s'}}\phi_\Theta$. Let
\[\phi:=\diam{a}\bigoplus_{s'\in\support{\Delta}}\Delta(s')\cdot{\downarrow}\!\!\!\bigwedge_{\sset{\Theta\in T\mid s'_{\Theta}=s'}}\phi_{\Theta}.\]
It is clear that $\pdist{s}\models\phi$, hence
$\pdist{t}\models\phi$ by $s \CR t$. It follows that
there must be a $\Theta^\ast$ with $\pdist{t}\ar{a}\Theta^\ast$,
$\Theta^\ast=\sum_{s'\in\support{\Delta}}\Delta(s')\cdot\Theta^\ast_{s'}$
and for each $s'\inp\support{\Delta},~ t'\inp\support{\Theta^\ast_{s'}}$ we have
$\pdist{t'}\models\bigwedge_{\sset{\Theta\in T\mid s'_\Theta =
s'}}\phi_\Theta$. This means that $\Theta^\ast\not\in T$ and hence
for each $s'\inp\support{\Delta},~ t'\inp\support{\Theta^\ast_{s'}}$ we
have $\pdist{s'}\equiv_\BISI \pdist{t'}$, i.e.\ $s' \mathbin{\CR} t'$.
Consequently, we obtain $\Delta\lift{\CR} \Theta^\ast\!\!$.\\ By symmetry
all transitions of $t$ can be matched up by transitions of $s$.
\hfill\qed
\end{proof}

Modal characterisation of strong and weak probabilistic bisimulation
has been studied in \cite{PS07}. It is also based on a
probabilistic extension of the Hennessy-Milner logic.
Instead of our modalities $\bigoplus$ and $\downarrow$ they use a
modality $[\cdot ]_p$. Intuitively, a
distribution $\Delta$ satisfies the formula $[\phi]_p$ when the set
of states satisfying $\phi$ is measured by $\Delta$ with probability
at least $p$. So the formula $[\phi]_p$ can be expressed by our
logics in terms of the probabilistic choice $\bigoplus_{i\in
I}p_i\cdot\phi_i$ by setting $I\mathbin=\{1,2\}$, $p_1\mathbin=p$, $p_2\mathbin=1\mathord-p$,
$\phi_1\mathbin={\downarrow}\phi$, and $\phi_2\mathbin=\true$.
Furthermore, instead of our modality $\diam{a}$, they use a modality
$\cdot\!\!\!\Diamond a$ that can be expressed in our logic by
$\cdot\!\!\!\Diamond a \phi = \diam{a}{\downarrow}\phi$. 
We conjecture that our modalities $\diam{a}$ and $\bigoplus$ cannot be
expressed in terms of the logic of \cite{PS07}, and that a logic of
that type is unsuitable for characterising forward simulation or
failure simulation.

When restricted to deterministic pLTSs
(i.e., for each state and for each action, there exists at most one
outgoing transition), probabilistic bisimulations can be
characterised by simpler forms of logics, as observed in
\cite{LS91,DEP98,PS07}.

\section{Concluding remarks}\label{s:cr}
We have considered characteristic equation systems consisting of
equations of the form $X_s = \phi_s$ where, for each refinement preorder we
have characterised, $\phi_s$ is displayed in Table~\ref{tb:ces}.
\begin{table}
\begin{center}
\begin{tabular}{@{}|c|c|r|@{}}
\hline \emph{preorder} &  $\phi_s$  & $X_\Delta$~~~~~~~~~~~~
\\[3pt] \hline\hline strong
prob.\ bis. &
$(\bigwedge_{s\arh{a}\Delta}\diam{a}X_{\Delta})\wedge
(\bigwedge_{a\in\Act_\tau}\boxm{a}\bigvee_{\pdist{s}\arh{a}\Delta}
X_{\Delta} )$
& $\bigoplus_{s\in\support{\Delta}}\Delta(s)\cdot {\downarrow}X_s$
 \\[3pt] \hline
strong prob.\ sim. &
$\bigwedge_{s\arh{a}\Delta}\diam{a}X_{\Delta}$
& $\bigoplus_{s\in\support{\Delta}}\Delta(s)\cdot {\downarrow}X_s$
\\[3pt] \hline\hline weak
prob.\ bis. &
$(\bigwedge_{s\arh{a}\Delta}\diam{a}X_{\Delta})\wedge
(\bigwedge_{a\in\Act_\tau}\boxm{a}\bigvee_{\pdist{s}\darhhat{a}\Delta}
X_{\Delta} )$
& $\bigoplus_{s\in\support{\Delta}}\Delta(s)\cdot {\downarrow}X_s$
 \\[3pt] \hline
weak prob.\ sim. &
$\bigwedge_{s\arh{a}\Delta}\diam{a}X_{\Delta}$  
& $\bigoplus_{s\in\support{\Delta}}\Delta(s)\cdot {\downarrow}X_s$
\\[3pt] \hline forward sim. &
$\bigwedge_{s\arh{a}\Delta}\diam{a}X_{\Delta}$
& $\bigoplus_{s\in\support{\Delta}}\Delta(s)\cdot X_s$
\\[3pt]
\hline failure sim.\ & \multicolumn{1}{l|}{$\left\{\begin{array}{@{}ll@{}}
\bigwedge_{s\arh{a}\Delta}\diam{a}X_{\Delta} & {\rm if}\ s\arh{\tau}\\[3pt]
(\bigwedge_{s\arh{a}\Delta}\diam{a}X_{\Delta})\wedge
\refuse{\setof{a\!}{\!s\nar{a}}}  & {\rm otherwise}
\end{array}\right.$}
& $\bigoplus_{s\in\support{\Delta}}\Delta(s)\cdot X_s$
\\[3pt] \hline
\end{tabular}
\end{center}
\caption{Characteristic equation systems $E: X_s = \phi_s$}\label{tb:ces}
\end{table}
Although they are in similar forms, the interpretations of formulae
$\diam{a}\phi$ and $\boxm{a}\phi$ change from the strong to the weak
case (Table~\ref{f:semantics}).

For the strong and weak probabilistic (bi)simulation, we could also have
used a state-based logic. To be precise, the modalities $\bigwedge$,
$\bigvee$, $\neg$, $\mu$ and $\nu$ would be interpreted on states
rather than distributions, $\bigoplus$ remains interpreted on
distributions, $\diam{a}$ and $\boxm{a}$ take a distribution-interpreted
formula as argument and return a state-interpreted formula, and
$\downarrow$ does just the reverse:
\begin{center}
$\begin{array}{|rcl@{\qquad\qquad\mbox{so}\qquad\qquad}rcl|}
\hline
&&\multicolumn{4}{l|}{~}\\[-8pt]
\Op{\diam{a}\phi}_\rho & = & \multicolumn{4}{l|}{ \setof{s\in S}{\exists
  \Delta':\pdist{s} \ar{a} \Delta'\ \wedge\ \Delta'\in\Op{\phi}_\rho} }\\
\Op{\boxm{a}\phi}_\rho & = &  \multicolumn{4}{l|}{\setof{s\in S}{\forall
  \Delta':\pdist{s} \ar{a} \Delta'\ \Rightarrow\ \Delta'\in\Op{\phi}_\rho}}\\
\Op{\downarrow\!\!\phi}_\rho & = & \multicolumn{4}{l|}{
\setof{\Delta\in\dist{S}}{\forall s\in\support{\Delta}: s\in\Op{\phi}_\rho}} \\[2pt]
\hline
\end{array}$
\end{center}
In fact, all our results and proofs are applicable to such a
state-based logic, with no significant change. Now a treatment of the
original strong bisimulation of \cite{HJ90} and the strong simulation of
\cite{JL91} proceeds exactly as this state-based treatment of strong
probabilistic (bi)simulation, but using $s$ rather than $\pdist{s}$ in
the definition of $\diam{a}$ and $\boxm{a}$.

There are many other behavioural relations studied in the
literature. It would be interesting to see if our approach of
deriving characteristic formulae applies to some of them. For
instance, probabilistic may and must testing preorders have a close
relationship with forward and failure simulations respectively
\cite{DGHM08}, so it appears promising to derive characteristic
formulae for them.

Another research direction is to exploit characteristic
formulae for deciding probabilistic behavioural relations and compare
it with other methods of deciding behavioural relations.

\vspace{-1em}
\section*{Acknowledgement}
\vspace{-1.5mm}
We thank Chenyi Zhang for interesting discussions about an early version of this paper.

\bibliographystyle{eptcs}
\bibliography{bibfile}

\end{document}